\documentclass{comsoc2008}
\usepackage{times}
\usepackage{amsmath}
\usepackage{amsfonts}
\usepackage{amsthm}

\newtheorem{definition}{Definition}
\newtheorem{theorem}[definition]{Theorem}
\newtheorem{lemma}[definition]{Lemma}

\theoremstyle{definition}

\title{Computing the nucleolus of weighted voting games}
\author{Edith Elkind and Dmitrii Pasechnik}

\newcommand\veceta{{\boldsymbol{\eta}}}
\newcommand\w{{\mathbf w}}
\newcommand\p{{\mathbf p}}
\newcommand\q{{\mathbf q}}
\newcommand\x{{\mathbf x}}

\renewcommand\d{{\mathbf d}}

\renewcommand\P{{\mathcal P}}
\renewcommand\S{{\mathcal S}}
\newcommand\LP{{\mathcal{LP}}}
\newcommand\eps{{\varepsilon}}

\newcommand{\aaai}{wvg07}
\newcommand{\gls}{MR1261419}
\newcommand{\schrlp}{MR874114}
\newcommand{\schmeidler}{MR0260432}
\newcommand{\matsui}{MR1768387}
\newcommand{\dengpapa}{deng:94a}
\newcommand{\deng}{MR2368807}
\newcommand{\cyclic}{MR2101908}
\newcommand{\nunes}{MR2109597}
\newcommand{\wkern}{MR1980665}
\newcommand{\tijs}{MR1952719}
\newcommand{\shapleyshubik}{MR989821}

\begin{document}

\begin{abstract}
Weighted voting games (WVG) are coalitional games 
in which an agent's contribution to a coalition is given by his {\it 
weight}, and a coalition wins if its total weight meets or exceeds 
a given quota. These games model decision-making in 
political bodies as well as collaboration and surplus division 
in multiagent domains. The computational complexity of 
various solution concepts for weighted voting games received 
a lot of attention in recent years. In particular, Elkind 
et al.(2007) studied the complexity of stability-related solution 
concepts in WVGs, namely, of the core, the least 
core, and the nucleolus. While they have completely characterized 
the algorithmic complexity of the core and the least 
core, for the nucleolus they have only provided an NP-hardness result.
In this paper, we solve an open problem posed by Elkind et al. by showing 
that the nucleolus of WVGs, and, more generally, 
$k$-vector weighted voting games with fixed $k$, can be computed in 
pseudopolynomial time, i.e., there exists an algorithm that 
correctly computes the nucleolus and runs in time 
polynomial in the number of players $n$ and the 
maximum weight $W$. In doing so, we propose a general framework for 
computing the nucleolus, which may be applicable to a wider of class of games.
 \end{abstract}


\section{Introduction}
Both in human societies and in multi-agent systems, 
there are many situations where individual agents can achieve their goals
more efficiently (or at all) by working together. This type of scenarios
is studied by {\em coalitional game theory}, which provides tools 
to decide which teams of agents will form and how they will divide the
resulting profit. In general, to describe a coalitional game, one has to specify 
the payoff available to every team, i.e., every possible subset of agents.
The size of such representation is exponential in the number of agents, 
and therefore working with a game given in such form is computationally 
intensive. For this reason, a lot of research effort has been spent on identifying
and studying classes of coalitional games that correspond to rich and practically 
interesting classes of problems and yet have a compact representation.

One such class of coalitional games is
{\em weighted voting games}, in which an 
agent's contribution to a coalition is given by his {\it weight}, 
and a coalition has value 1 if its total weight meets or exceeds a given quota, 
and 0 otherwise. These games model decision-making in political bodies, 
where agents correspond to political parties 
and the weight of each party is the number of its supporters, 
as well as task allocation in multi-agent systems, 
where the weight of each agent is the amount of resources
it brings to the table and the quota is the total amount of resources
needed to execute a task.

An important issue in coalitional games is {\it surplus division}, 
i.e., distributing the value of the resulting coalition
between its members in a manner that encourages cooperation. 
In particular, it may be 
desirable that all agents work together, i.e., form the 
{\em grand coalition}. In this case, a natural goal is to distribute 
the payoff of the grand coalition so that it remains {\em stable}, 
i.e., so as to minimize the incentive 
for groups of agents to deviate and form coalitions of their own.
Formally, this intuition is captured by several related solution
concepts, such as the core, the least core, and the nucleolus.
Without going into the technical details of their definitions 
(see Section~\ref{sec:prelim}), the nucleolus is, in some sense, 
the most stable payoff allocation scheme, and as such it 
is particularly desirable when the stability of the grand coalition is important.

The stability-related solution concepts for WVGs
have been studied 
from computational perspective in~\cite{\aaai}. There, 
the authors show that while computing the core is easy, finding the least
core and the nucleolus is NP-hard. These computational hardness results
rely on all weights being given in binary, which suggests that these 
problems may be easier for polynomially bounded weights.
Indeed, paper~\cite{\aaai} provides a pseudopolynomial time
algorithm (i.e., an algorithm whose running time is polynomial
in the number of players $n$ and the maximal weight $W$) for the least core. 
However, an analogous question for the nucleolus has been left open.

In this paper, we answer this question in affirmative 
by presenting a pseudopolynomial
time algorithm for computing the nucleolus.
\begin{theorem}\label{thm:main}
For a WVG specified by integer weights $w_1,\dots,w_n$ and a quota $q$, 
there exists a procedure
that computes its nucleolus in time polynomial in $n$ and $W=\max_i w_i$.
\end{theorem}
As in many practical scenarios
(such as e.g., decision-making in political bodies) the weights are likely 
to be not too large, this provides a viable algorithmic solution 
to the problem of finding the nucleolus. Our approach relies
on solving successive exponential-sized linear programs by constructing
dynamic-programming based separation oracles, a technique that may prove
useful in other applications.

A proof of Theorem~\ref{thm:main} is presented in Section~\ref{sec:algo}, 
after preliminaries in Section~\ref{sec:prelim} and a discussion of related 
work in Section~\ref{sec:relwork}. The text rounds up by 
Section~\ref{sec:conclusion} that discusses conclusions and future work
directions.

\section{Related work}\label{sec:relwork}
Another approach to payoff distribution in weighted voting games is based
on {\em fairness}, i.e., dividing the payoff in a manner that is proportional
to the agent's influence. The most popular solution concepts used
in this context are the Shapley--Shubik power index~\cite{\shapleyshubik} 
and the Banzhaf
power index~\cite{banzhaf}. 
Both of these indices are known to be computationally hard
for large weights~\cite{prasadkelly,\dengpapa}, 
yet efficiently computable for polynomially bounded weights~\cite{\matsui}.

The concept of the nucleolus was introduced by Schmeidler~\cite{\schmeidler}
in 1969. Paper~\cite{\schmeidler} explains how the nucleolus arises naturally
as ``the most stable'' payoff division scheme, 
and proves that the nucleolus is well-defined for any coalitional 
game and is unique. 
Kopelowitz~\cite{Kopel67} proposes to compute the nucleolus 
by solving a sequence of linear programs; we use this approach in our algorithm.

The computational complexity of the nucleolus has been studied 
for many classes of games, such as flow games~\cite{\deng}, 
cyclic permutation games~\cite{\cyclic}, assignment games~\cite{\nunes}, 
matching games~\cite{\wkern}, and neighbor games~\cite{\tijs}, 
as well as several others. While some of these papers provide polynomial-time
algorithms for computing the nucleolus, others contain NP-hardness results.

The work in this paper is inspired by~\cite{\aaai}, which 
shows that the least core and the nucleolus of a weighted voting game
are NP-hard to compute. It also proves that
the nucleolus cannot be approximated within any constant factor.
On the positive side, it provides a pseudopolynomial time algorithm
for computing the least core, i.e., an algorithm whose running time is polynomial
in $n$ and $W$ (rather than in the game representation size $O(n\log W$)), as well
as a fully polynomial time approximation scheme (FPTAS) for the least core.
However, for the nucleolus, paper~\cite{\aaai} contains no algorithmic results.

\section{Preliminaries and Notation}\label{sec:prelim}
A {\em coalitional game} $G=(I, \nu)$ is given by a set of agents
$I=\{1, \dots, n\}$
and a function $\nu : 2^I \to \mathbb{R}$
that maps any subset (coalition) of the agents
to a real value. This value is the total utility these agents
can guarantee to themselves when working together. 
A coalitional game is called {\em simple} if 
$\nu(S)\in\{0, 1\}$ for any coalition $S\subseteq I$.
In a simple game, a coalition $S$ is called {\em winning}
if $\nu(S)=1$ and {\em losing} otherwise.

A {\em weighted voting game} is a simple coalitional game $G$ given by a
set of agents $I=\{1, \dots, n\}$, their non-negative {\em weights}
$\w=(w_1, \dots, w_n)$, and a {\em quota} $q$; we write
$G=(I;\w; q)$.
As the focus of this paper is computational complexity
of such games, it is important to specify how the game is represented.
In what follows, we assume that the weights and the quota are integers 
given in binary. This does not restrict the class of WVGs that we can work 
with, as any weighted voting game has such a representation~\cite{muroga:71a}. 

For a coalition $S\subseteq I$, its value $\nu(S)$ is
1 (i.e., $S$ is winning)
if $\sum_{i\in S}w_i\ge q$; otherwise, $\nu(S)=0$. 
Without loss of
generality, we assume that the value of the grand coalition $I$ is 1, 
that is, $\sum_{i\in I}w_i\ge q$.
Also, we set $W=\max_{i\in I}w_i$.

For a coalitional game $G=(I, \nu)$, an {\em imputation} is a vector of
non-negative numbers $\p=(p_1, \dots, p_n)$, one for each agent
in $I$, such that $\sum_{i\in I}p_i=\nu(I)$.  We refer to $p_i$ as the
\emph{payoff} of agent~$i$.  We write $p(S)$ to denote $\sum_{i\in
S}p_i$. Similarly, $w(S)$ denotes $\sum_{i\in S}w_i$.

An important notion in coalitional games is that of stability:
intuitively, a payoff vector should distribute the gains of the grand coalition 
in such a way that no group of agents has an incentive to deviate and form a 
coalition of their own. This intuition is captured by the notion of the core:
the {\em core} of a game $G$ is the set of all imputations $\p$ such 
that 
\begin{equation}\label{core}
p(S)\ge\nu(S) \text{ for all } S\subseteq I.
\end{equation}
While the core is an appealing solution concept, it is very demanding:
indeed, for many games of interest, the core is empty. In particular, 
it is well known that in simple games 
the core is empty unless there exists a veto player, i.e., a player that   
is present in all winning coalitions. Clearly, this is not always the case
in weighted voting games, and a weaker solution concept is needed.

We can relax the notion of the core by
allowing a small error in the inequalities~(\ref{core}).  
This leads to the notion of $\eps$-core: 
the {\em $\eps$-core} of a game $G$ is the set of all imputations $\p$ 
such that $p(S)\ge\nu(S)-\eps$  for all  $S\subseteq I$.
Under an imputation in the $\eps$-core, 
the {\em deficit} of any coalition $S$, 
i.e., the difference $\nu(S)-p(S)$
between its value and the payoff that it gets, 
is at most $\eps$. Observe that if $\eps$ is large enough, e.g., 
$\eps\ge 1$, then the $\eps$-core is guaranteed to be non-empty.
Therefore, a natural goal is to identify
the smallest value of $\eps$ such that the $\eps$-core
is non-empty, i.e., to minimize the error introduced 
by relaxing the inequalities in~(\ref{core}).
This is captured by the concept
of the least core, defined as the smallest non-empty $\eps$-core
of the game.
More formally, consider the set 
$\{\eps\mid \eps\le 1, \eps\text{-core of $G$ is non-empty}\}$. 
It is easy to see that this 
set is compact, so it has a minimal element $\eps_1$. 
The {\em least core} of $G$ is its $\eps_1$-core.
The imputations in least core distribute the payoff in a way
that minimizes the incentive to deviate: under any $\p$
in the least core,  
no coalition can gain more than $\eps_1$ by deviating, 
and for any $\eps'<\eps_1$, 
there is no way to distribute the payoffs so that 
the deficit of every coalition is at most $\eps'$.
However, while the least core
minimizes the worst-case deficit, it does not attempt
to minimize the {\em number} of coalitions that experience
the worst deficit, i.e., $\eps_1$, nor does it try
to minimize the second-worst deficit, etc. The {\em nucleolus}
is a refinement of the least core that takes into account these 
higher-order effects. 

Recall that the deficit of a coalition $S$ under an imputation $\p$
is given by $d(\p, S)=\nu(S)-p(S)$. The {\em deficit vector}
of $\p$ is the vector
$\d(\p) = (d(\p,S_1), \dots, d(\p, S_{2^n}))$,
where $S_1, \dots, S_{2^n}$ is a list
of all subsets of $I$ ordered so that $d(\p, S_1)\ge d(\p,
S_2)\ge\dots\ge d(\p, S_{2^n})$. In other words, the deficit vector
lists the deficits of all coalitions from the largest to the smallest 
(which may be negative). The {\em nucleolus} is an imputation
$\veceta=(\eta_1, \dots, \eta_n)$ that satisfies 
$\d(\veceta)\le_{\mathrm{lex}}
\d(\x)$ for any other imputation $\x$, where $\le_{\mathrm{lex}}$
is the lexicographic order. It is known~\cite{\schmeidler}
that the nucleolus is well-defined
(i.e., an imputation with a lexicographically
minimal deficit vector always exists) and is unique.

\section{Algorithm}\label{sec:algo}
The description of our algorithm is structured as follows. 
We use the idea of~\cite{Kopel67}, 
which explains how to compute the nucleolus by solving a
sequence of (exponential-size) linear programs. In Section~\ref{sequence-lp}, 
we present the approach of~\cite{Kopel67}, 
and argue that it correctly computes the nucleolus. This material is not new, 
and is presented here for completeness.
In Section~\ref{counts}, we show how to design separation oracles
for the linear programs in this sequence so as to solve them 
by the ellipsoid method.
While a naive implementation of these separation oracles would require
storing exponentially many constraints, we show how to replace
explicit enumeration of these constraints with a counting subroutine, 
while preserving
the correctness of the algorithm. 
The arguments
in Sections~\ref{sequence-lp} and~\ref{counts}
apply to {\em any} coalitional game
rather than just weighted voting games.

In Section~\ref{dyn-prog}, we show that for weighted voting games
with polynomially-bounded weights the counting subroutine used by 
the algorithm of Section~\ref{counts} can be efficiently implemented.
Finally, in Section~\ref{violated-constraint} we show
how to modify this subroutine 
to efficiently identify a violated constraint
if a given candidate solution is infeasible.  
The results in Sections~\ref{dyn-prog} and~\ref{violated-constraint}
are specific to weighted voting games with polynomially bounded weights.

\subsection{Computing the nucleolus by solving successive
linear programs}\label{sequence-lp}
As argued in \cite{Kopel67}, the nucleolus can be computed by solving at most
$n$ successive linear programs.
The first linear program $\LP^1$ contains the inequality 
$p(S)\ge\nu(S)-\eps$ for each coalition $S\subseteq I$,
and attempts to minimize $\eps$ subject to these inequalities, i.e., it computes
a payoff in the least core as well as the value $\eps^1$ of the least core.
Given a {\em {\rm (}relative{\rm )} interior optimizer} 
$(\p^1, \eps^1)$ for $\LP^1$
(i.e., an
optimal solution that minimizes the number of tight constraints), 
let $\Sigma^1$ be the set of all inequalities in $\LP^1$
that have been made tight by $\p^1$ (we will abuse notation and use $\Sigma^1$
to refer both to these inequalities and the corresponding coalitions).
We construct
the second linear program $\LP^2$ by replacing all inequalities in $\Sigma^1$
with equations of the form
$p(S)=\nu(S)-\eps^1$, and try to minimize $\eps$ subject to this new
set of constraints. This results in $\eps^2<\eps^1$ and a payoff
vector $\p^2$ that satisfies $\p^2(S)=\nu(S)-\eps^1$ for all $S\in\Sigma^1$,
$\p^2(S)\ge\nu(S)-\eps^2$ for all $S\not\in\Sigma^1$.
We repeat this process until the payoffs to all coalitions are determined, 
i.e., the solution space of the current linear program consists of a single point.
It has been shown~\cite{Kopel67} 
that this will happen after at most $n$ iterations:
indeed, each iteration reduces the dimension of the solution space by at least 1.

More formally, the sequence of linear programs $(\LP^1, \dots, \LP^n)$
is defined as follows. 
The first linear program $\LP^1$ is given by
\begin{equation}\label{LP1}
\min_{(\p,\eps)} \eps\quad\text{subject to}\quad\left\{
\begin{aligned}
\sum_{i\in I}p_i&=1,\quad p_i\ge 0 \quad\text{ for all } i=1, \dots, n\\
\sum_{i\in S}p_i&\ge \nu(S)-\eps \text{ for all } S\subseteq I.
\end{aligned}\right.
\end{equation}
Let $(\p^1, \eps^1)$ be an interior optimizer to this linear program.
Let $\Sigma^1$ be the set of tight constraints for $(\p^1, \eps^1)$ 
(and, by a slight abuse of 
notation, the coalitions that correspond to them), i.e., 
for any $S\in\Sigma^1$ we have $p^1(S)=1-\eps^1$.

Now, suppose that we have defined the first $j-1$ linear programs
$\LP^1, \dots, \LP^{j-1}$. For $k=1, \dots, j-1$, 
let $(\p^k, \eps^k)$ be an interior optimizer for $\LP^k$
and let $\Sigma^k=\{S\mid p^k(S)=\nu(S)-\eps^k\}$.
Then the $j$th linear program $\LP^{j}$ is given by
\begin{equation}\label{LPi+1}
\min_{(\p,\eps)}\eps\quad\text{subject to}\quad\left\{
\begin{aligned}
\sum_{i\in I}p_i&=1,\quad p_i\ge 0\quad \text{ for all } i=1, \dots, n\\
\sum_{i\in S}p_i&= \nu(S)-\eps^1\text{ for all } S\in\Sigma^1\\
&\dots\\
\sum_{i\in S}p_i&= 
\nu(S)-\eps^{j-1}\text{ for all } S\in\Sigma^{j-1}\\
\sum_{i\in S}p_i&\geq \nu(S)-\eps\quad\text{ for all } 
S\not\in\cup_{k=1}^{j-1}\Sigma^k.
\end{aligned}
\right.
\end{equation}
Fix the minimal value of $t$ such that 
there is no interior solution to $\LP^t$.
It is not hard to see that the (unique) solution to $\LP^t$
is indeed the nucleolus. Indeed, the nucleolus is a payoff vector that
produces the lexicographically maximal deficit vector. This means that it:
\begin{itemize}
\item[(i)] 
minimizes $\eps^1$ such that all coalitions receive at least $1-\eps^1$;
\item[(ii)]
given (i), 
minimizes the number of coalitions that receive $1-\eps^1$;
\item[(iii)]
given (i) and (ii), minimizes $\eps^2$ such that all coalitions except for 
those receiving $1-\eps^1$ receive at least $1-\eps^2$;
\item[(iv)] 
given (i), (ii) and (iii), minimizes the number of coalitions that
receive $1-\eps^2$, etc.
\end{itemize}
Our sequence of linear programs finds a payoff vector that satisfies
all these conditions; in particular,  (ii) and (iv) (and analogous 
conditions at subsequent steps) are satisfied, since at each step we choose 
an interior optimizer for the corresponding linear program. 
The only issue that we have to address is that 
our procedure selects an {\em arbitrary} interior optimizer to the 
current linear program in order to construct the set $\Sigma^j$.
Conceivably, this may have an impact on the final solution: if two interior optimizers 
to $\LP^j$ lead to two different sets $\Sigma^j$, they may also result in different 
values of $\eps^{j+1}$, so one would have to worry about choosing the {\em right}
interior optimizer. Fortunately, this is not the case, as shown by the following lemma.
\begin{lemma}\label{lem}
Any two interior optimizers $(\p, \eps)$ and $(\q, \eps)$ 
for the linear program $\LP^j$ have the same set of tight constraints, 
i.e., the set $\Sigma^j$ is independent of the choice of the interior optimizer.
\end{lemma} 
\begin{proof}
First note that the set of all interior optimizers for $\LP^j$
is convex. Now, suppose that  $\p$ and $\q$ are two interior optimizers
for $\LP^j$, but have different sets of tight constraints. Then, by convexity, 
any convex combination $\alpha\p+(1-\alpha)\q$ of $\p$ and $\q$
is also an interior optimizer for $\LP^j$. However, the set of constraints
that are tight for $\alpha\p+(1-\alpha)\q$ is the intersection 
of the corresponding sets for $\p$ and $\q$, i.e., 
$\alpha\p+(1-\alpha)\q$ has strictly fewer tight constraints than $\p$
or $\q$, a contradiction with $\p$ and $\q$ being interior optimizers for $\LP^j$.
\end{proof}

We conclude that when this algorithm terminates, the output is indeed
the nucleolus. Next, we discuss how to solve
each of the linear programs $\LP^j$, $j=1, \dots, t$.

\subsection{Solving the linear programs $\LP^1, \dots, \LP^t$}\label{counts}
It is well-known (see e.g. \cite{\schrlp,\gls}) 
that a linear program can be solved in polynomial time
by the ellipsoid method
as long as it has a polynomial-time {\em separation oracle}, 
i.e., an algorithm that, given a 
candidate feasible solution, either confirms that it is feasible or outputs a 
violated 
constraint. Moreover, the ellipsoid method can also be used to find an interior
optimizer (rather than an arbitrary optimal solution) in polynomial 
time~\cite[Thm.~6.5.5]{\gls}, 
as well as to decide whether one exists~\cite[Thm.~6.5.6]{\gls}.
We will now construct a polynomial-time separation oracle for $j$th linear
program $\LP^j$ in our sequence. 

It is easy to construct the part responsible for
checking equations of $\LP^{j}$ in~\eqref{LPi+1}, assuming that
we already have an oracle for the $(j-1)$st program $\LP^{j-1}$.
Indeed, the latter oracle can be easily modified to also check
whether the equation $\eps=\eps^{j-1}$ holds, thus providing
an oracle for the optimal face of $\LP^{j-1}$. 
Then by~\cite[Thm.~6.5.5]{\gls} we can compute a basis of the optimal
face (which consists of at most $n$ equations) in polynomial time.
The separation oracle can then reject a candidate solution $(\p, \eps)$
if $\p$ violates one of those basis equations.

Dealing with the inequalities of $\LP^{j}$ is more complicated. 
A naive separation oracle would have to explicitly list the sets 
$\Sigma^1, \dots, \Sigma^{j-1}$, which may be superpolynomial in size.
Alternatively, one can treat this part of
the oracle as a 0-1 integer linear feasibility problem, with 0-1 variables
$x_i$ encoding a set $S\not\in\cup_{k=1}^{j-1}\Sigma^k$ 
that provides a separating inequality for the oracle input
$(\p,\eps)$. Namely, suppose that we have verified
that $\p$ satisfies all the equations in $\LP^j$ (as described above).
Then, given an interior optimizer 
$(\p^{j-1}, \eps^{j-1})$ for $\LP^{j-1}$, the values $x_1, \dots, x_n$ 
can be obtained as a solution to the following inequalities:
\begin{align}
\label{eq:pjminone} \sum_i p^{j-1}_i x_i& > 1-\eps^{j-1},\\ 
\label{eq:p} \sum_i p_i x_i& < 1-\eps,\\ 
\sum_i w_i x_i&\geq q.
\end{align}
The problem with this approach is that for arbitrary
rational $\p$ and $\p^{j-1}$ this system of inequalities
is at least as hard as {\sc Knapsack}, which is NP-complete.
Moreover, as $\p$ and $\p^{j-1}$ are produced by the ellipsoid method, 
there is no guarantee that their bitsizes
are small enough to use a (pseudo)polynomial-time
algorithm for {\sc Knapsack}.
The only hope is to replace at least one of \eqref{eq:pjminone} and
\eqref{eq:p} by something ``tame''.  

We will now present a more sophisticated approach 
to identifying a violated constraint. In a way, it can be seen as
replacing checking~(\ref{eq:pjminone}) with counting.
Our construction proceeds by induction:
to construct a separation oracle for $\LP^j$, we assume that we have constructed
an oracle for $\LP^{j-1}$, and are given the sizes of sets
$\Sigma^1, \dots, \Sigma^{j-1}$ as well as
the sequence $(\eps^1, \dots, \eps^{j-1})$.

By construction, any optimal solution $(\p, \eps)$ to $\LP^j$ satisfies
$\eps<\eps^{j-1}$, so we can add the constraint $\eps\le\eps^{j-1}$ to $\LP^j$
without changing the set of solutions. From now on, we will assume that
$\LP^j$ includes this constraint. Our separation oracle
will first check whether a given candidate solution $(\p, \eps)$
satisfies $\eps\le\eps^{j-1}$, 
as well as constraints $p_i\ge 0$ for all $i=1, \dots, n$
and $p(I)=1$, and reject $(\p, \eps)$ and output a violated constraint 
if this is not the case. Therefore, 
in what follows we assume that $(\p, \eps)$ satisfies all these 
easy-to-identify constraints. 

Now, a candidate solution $(\p, \eps)$ is feasible for $\LP^j$
if $p(S)=\nu(S)-\eps^t$ for $S\in\Sigma^t$, $t=1, \dots, j-1$,
and $p(S)\ge\nu(S)-\eps$ for all $S\not\in\cup_{t=1}^{j-1}\Sigma^t$.
Recall that the deficit of a coalition $S\subseteq I$
under a payoff vector $\p$ is given by $\nu(S)-p(S)$.
Suppose that we have a procedure $\P(\p, \eps)$
that, given a candidate solution $(\p, \eps)$, 
can efficiently compute the top $j$ 
distinct deficits under $\p$, i.e., 
\begin{align*}
m^1&=\max\{d(S)\mid S\subseteq I\}\\ 
m^2&=\max\{d(S)\mid S\subseteq I, d(S)\neq m^1\}\\ 
&\dots\\
m^j&=\max\{d(S)\mid S\subseteq I, d(S)\neq m^1, \dots, m^{j-1}\}
\end{align*}
as well as
the numbers $n^1, \dots, n^j$ of coalitions that have deficits of 
$m^1, \dots, m^j$, respectively:
\begin{equation*}
n^k=|\{S\mid S\subseteq I, d(S)=m^k\}|,\qquad k=1,\dots,j.  
\end{equation*}
Suppose also that we are given
the values $\eps^1, \dots, \eps^{j-1}$
and the sizes $s^t$
of the sets $\Sigma^t$, $t=1, \dots, j-1$.

Now, our algorithm works as follows. Given a candidate solution
$(\p, \eps)$, it runs 
$\P(\p, \eps)$ to obtain 
$m^t, n^t$, $t=1, \dots, j$.
If $\eps<\eps^{j-1}$, it then checks whether
\begin{itemize}
\item[(a)]
$m^t=\eps^t$ and $n^t=s^t$ for all $t=1, \dots, j-1$ 
\item[(b)]
$m^j\le \eps$. 
\end{itemize}
If $\eps=\eps^{j-1}$, it
simply checks whether $m^t=\eps^t$ for $t=1, \dots, j-1$
and $n^t=s^t$ for all $t=1, \dots, j-2$
\footnote{Alternatively, if $\eps=\eps^j$, 
one can verify whether $(\p, \eps)$ is a feasible solution to the previous linear
program $\LP^{j-1}$.}.
If these conditions are satisfied, the algorithm answers
that $(\p, \eps)$ is indeed a feasible solution, and otherwise it 
identifies and outputs a violated constraint (for details of this step, see 
Section~\ref{violated-constraint}).
We will now show that this algorithm implements a separation oracle 
for $\LP^j$ correctly and efficiently. 
\begin{theorem}\label{th:oracleisok}
Given the values $\eps^t, s^t$, $t=1, \dots, j-1$, and a procedure 
$\P(\p, \eps)$ that
computes $m^t, n^t$, $t=1, \dots, j$, in polynomial time, our algorithm
correctly decides whether a given pair $(\p, \eps)$ is feasible for 
$\LP^j$ and runs in polynomial time.
\end{theorem}
\begin{proof}
We start by proving an auxiliary lemma.
\begin{lemma}\label{induction}
For any vector $\p$ and any $t\le j-1$
such that $m^s=\eps^s$, $n^s=s^s$ for all $s\le t$,
the coalitions with deficit $\eps^s$
under $\p$
are exactly the ones in $\Sigma^s$.
\end{lemma}
\begin{proof}
The proof is by induction on $s$.
For $s=1$, we have
that the largest deficit of any coalition
under $\p$ is $\eps^1$, and there are
exactly $s^1$ coalitions with this deficit. Hence,
$(\p, \eps^1)$ is an interior optimizer for $\LP^1$,
and therefore the lemma follows by Lemma~\ref{lem}.
Now, suppose that the lemma has been proven for $s-1$.
By the induction hypothesis, $\p$
satisfies all constraints in $\Sigma^1, \dots, \Sigma^{s-1}$.
Also, under $\p$ there are at most $s^1+\dots+s^{s-1}$
coalitions whose deficit exceeds $\eps^{s}$, so for all coalitions
not in $\cup_{r=1}^{s-1}\Sigma^{r}$ their deficit is at most $\eps^s$.
Finally, there are exactly $s^s$ coalitions whose deficit is exactly
$\eps^s$. Hence,
$(\p, \eps^s)$ is an interior optimizer for $\LP^s$,
and therefore the lemma follows by Lemma~\ref{lem}.
\end{proof}
To prove the theorem, let us first consider the case $\eps<\eps^{j-1}$.
Suppose that $(\p, \eps)$ satisfies (a) and (b).
By using Lemma~\ref{induction} with $t=j-1$, we conclude
that $(\p, \eps)$ satisfies all equations in $\LP^j$. Now,
under $\p$,
$m^j$ is the largest deficit of a coalition not in $\cup_{t=1}^{j-1}\Sigma^t$.
If this deficit is at most $\eps$, then the pair $(\p, \eps)$ is a feasible
solution to $\LP^j$.

Conversely, suppose that (a) or (b) is violated.
If $(\p, \eps)$ satisfies (a) but violates (b), by using Lemma~\ref{induction}
with $t=j-1$ we conclude that,
under $\p$ the coalitions in $\Sigma^t$ have deficit $\eps^t$
for  $t=1, \dots, j-1$,
but the deficit of some coalition not in $\cup_{t=1}^{j-1}\Sigma^t$ exceeds
$\eps$. Hence, this coalition corresponds to a violated constraint.
Now, suppose that (a) does not hold,
and let $s$ be the smallest index for which
$m^{s}\neq \eps^{s}$ or $n^s\neq s^s$. By using Lemma~\ref{induction}
with $t=s-1$, we conclude that for $r=1, \dots, s-1$
the coalitions in $\Sigma^{r}$ have deficit $\eps^{r}$.
However, either the $s$th distinct deficit under $\p$
is not $\eps^s$, in which case $\p$ violates a constraint in
$2^I\setminus\cup_{r=1}^{s-1}\Sigma^r$,
or under $\p$ there are more than $s^s$ coalitions
with deficit $\eps^s$ (note that by construction
it cannot be the case that $n^s<s^s$). In the latter case, there is a coalition
in $2^I\setminus\cup_{r=1}^s\Sigma^{r}$ whose
deficit exceeds $\eps^s$,
thus violating the corresponding constraint.
The case $\eps=\eps^{j-1}$ is similar. In this case for a candidate solution
$(\p, \eps)$ to be feasible, it is not required that
there are exactly $s^{j-1}$ coalitions with
deficit $\eps^{j-1}$. Hence, the algorithm only has to decide if
for all $t=1, \dots, j-2$, the coalitions with deficit $\eps^t$
under $\p$ are exactly the ones in $\Sigma^t$, and all other coalitions
get at least $\eps^{j-1}$. Showing that our algorithm checks
this correctly can be done similarly to the previous case.

The bound on the running time is obvious from the description of the algorithm.
\end{proof}

To provide the value $s^j=|\Sigma^j|$ for the subsequent linear programs
$\LP^{j+1}, \dots, \LP^n$, we need to find an interior optimizer for $\LP^j$. 
Thm.~6.5.5 in~\cite{\gls} explains how to do this given 
a separation oracle for the optimal face, i.e., the set of all optimizers of $\LP^j$.
Observe that such an oracle can be obtained by a slight modification
of the oracle described above. Indeed, the optimal face is the set of solutions
to the linear feasibility problem given by the constraints in $\LP^j$
together with the constraint $\eps=\eps^j$. The modified oracle 
first checks the latter constraint, reports the violation (and the corresponding
inequality) if it happens, and otherwise continues as the original oracle. 
Clearly, the modified oracle runs in polynomial time whenever the original one does.
Hence, we can compute $s^j$ in polynomial time by
computing an interior solution 
$(\p, \eps)$ to $\LP^j$ according to~\cite[Thm.~6.5.5]{\gls},
running $\P(\p, \eps)$ to find $n^j$, and setting $s^j=n^j$.

\subsection{Implementing the counting}\label{dyn-prog}
We will now show how to implement the counting procedure $\P(\p, \eps)$
used in Section~\ref{counts} for WVGs.
The running time of our procedure is polynomial
in the number of players $n$ and the maximum weight $W$.

Our approach is based on dynamic programming. Fix 
a WVG $(I; \w; q)$, 
a payoff vector $\p$, and $j\le n$. 
For all $k=1, \dots, n$, $w=1, \dots, nW$, 
let $X^1_{k, w}, \dots, X^j_{k, w}$ be the bottom $j$
distinct payoffs to coalitions in $\{1, \dots, k\}$ of weight $w$, 
i.e., define 
\begin{align*}
X^1_{k, w} &= \min\{p(S)\mid S\subseteq \{1, \dots, k\}, w(S)=w \}\\
X^2_{k, w} &= \min\{p(S)\mid S\subseteq \{1, \dots, k\}, w(S)=w,
p(S)\neq X^1_{k, w}\}\\
&\dots\\
X^j_{k, w} &= \min\{p(S)\mid S\subseteq \{1, \dots, k\}, w(S)=w, 
p(S)\neq X^1_{k, w}, \dots, X^{j-1}_{k, w}\}
\end{align*}
and
let $Y^1_{k, w}, \dots, Y^j_{k, w}$ be the numbers of coalitions
that get these payoffs, i.e. set 
\begin{equation*}
Y^t_{k, w} = |\{S\mid S\subseteq \{1, \dots, k\}, w(S)=w,
p(S)=X^t_{k, w}\}|,\qquad
\text{for } t=1, \dots, j. 
\end{equation*}
These quantities can be computed inductively for $k=1, \dots, n$ as 
follows.

For $k=1$, we have
$X^1_{1, w} = p_1$ if $w=w_1$ and $+\infty$ otherwise,
$Y^1_{1, w} = 1$ if $w=w_1$ and $0$ otherwise, and
$X^t_{1, w} = +\infty$, $Y^t_{1, w}=0$ for $t=2, \dots, j$.

Now, suppose that we have computed
$X^1_{k-1, w}, \dots, X^j_{k-1, w}, Y^1_{k-1, w}, \dots, Y^j_{k-1, w}$
for all $w=1, \dots, nW$.
Consider $S\subseteq\{1, \dots, k\}$ receiving
one of the bottom $j$ distinct payoffs to subsets of $\{1, \dots, k\}$
of weight $w$, i.e., $p(S)\in\{X^1_{k, w}, \dots, X^j_{k, w}\}$.
Then either
\begin{itemize}
\item[(1)]
$S\subseteq\{1, \dots, k-1\}$, in which case
$S$ must be among the coalitions that receive
one of the bottom $j$ distinct payoffs to subsets of $\{1, \dots, k-1\}$
of weight $w$, i.e.,
we have $p(S)\in\{X^1_{k-1, w}, \dots, X^j_{k-1, w}\}$, or
\item[(2)]
$k\in S$, in which case $S\setminus\{k\}$ must be
among the coalitions that receive
one of the bottom $j$ distinct payoffs to subsets of $\{1, \dots, k-1\}$
of weight $w-w_k$, i.e.,
we have $p(S\setminus\{k\})\in\{X^1_{k-1, w-w_k}, \dots, X^j_{k-1, w-w_k}\}$.
\end{itemize}
Consider the multi-set
$\S_{k, w}=\{X^1_{k-1, w}, \dots, X^j_{k-1, w},
p_k+X^1_{k-1, w-w_k}, \dots, p_k+X^j_{k-1, w-w_k}\}$.
By the argument above, we have
\begin{align*}
X^1_{k, w} &= \min\{x\mid x\in\S_{k, w}\}\\
X^2_{k, w} &= \min\{x\mid x\in\S_{k, w}, x\neq X^1_{k, w}\}\\
&\dots\\
X^j_{k, w} &=
\min\{x\mid x\in\S_{k, w}, x\neq X^1_{k, w}, \dots, X^{j-1}_{k, w}\}.
\end{align*}
The number of coalitions that receive the payoff
$X^t_{k, w}$, i.e., $Y^t_{k, w}$, $t=1, \dots, j$, depends on how
many times $X^t_{k, w}$ appears in $\S_{k, w}$. If it only appears once,
then there is only one source of sets that receive a payoff of $X^t_{k, w}$,
i.e., we set $Y^t_{k, w}=Y^s_{k-1, w}$ if $X^t_{k, w}$ appears as
$X^s_{k-1, w}$ for some $s=1, \dots, j$, and we set
$Y^t_{k, w}=Y^s_{k-1, w-w_k}$ if $X^t_{k, w}$ appears as $X^s_{k-1, w-w_k}+p_k$
for some $s=1, \dots, j$. On the other hand,
if $X^t_{k, w}$ appears twice in $\S_{k, w}$ (first time as $X^s_{k-1, w}$
and second time as $p_k+X^{s'}_{k-1, w-w_k}$ for some $s, s'=1, \dots, j$),
we have to add up the corresponding counts, i.e., we set
$Y^t_{k, w}=Y^s_{k-1, w}+ Y^{s'}_{k-1, w-w_k}$.

After all $X^1_{n, w}, \dots, X^j_{n, w}, Y^1_{n, w}, \dots, Y^j_{n, w}$
have been evaluated, it is not hard to compute $m^t, n^t$, $t=1, \dots, j$.
Indeed, the top $j$ deficits
appear in the multi-set 
$$
\S=
\{I_w-X^1_{n, w}, \dots, I_w-X^j_{n, w}\mid w=1, \dots, nW\}, 
$$
where $I_w=1$ if $w\ge q$ and $I_w=0$ if $w<q$ (recall that $q$
is the quota of the game, i.e., $\nu(S)=1$ if and only if $w(S)\ge q$).
Hence, we can set
\begin{align*}
m^1 &= \max\{x\mid x\in\S\}\\
m^2 &= \max\{x\mid x\in\S, x\neq m^1\}\\
&\dots\\
m^j &= \max\{x\mid x\in\S, x\neq m^1, \dots, m^{j-1}\}.
\end{align*}
The procedure for computing $n^t$, $t=1, \dots, j$, is similar
to that of computing $Y^s_{k, w}$ (see above):
we have to check how many times $m^t$ appears in $\S$
and add the corresponding counts.

In the next subsection, 
we will show how to find a violated inequality if
$(\p, \eps)$ is not a feasible solution to $\LP^j$.

\subsection{Identifying a violated constraint}\label{violated-constraint}
Consider $\LP^j$ and a candidate solution $(\p, \eps)$.
Suppose that the algorithm described in the previous subsection
has decided that $(\p, \eps)$ is not a feasible solution to $\LP^j$.
This can happen in three possible ways.
\begin{itemize}
\item[(a)] 
$m^s=\eps^s$, $n^s=s^s$ for $s=1, \dots, \ell-1$, but
$m^\ell\neq \eps^\ell$ 
for an $\ell<j$.
\label{fcase}
\item[(b)] 
$m^s=\eps^s$, $n^s=s^s$ for $s=1, \dots, \ell-1$, $m^\ell=\eps^\ell$, 
but $n^\ell\neq s^\ell$
for an $\ell<j$.
\label{scase}
\item[(c)]
$m^s=\eps^s$, $n^s=s^s$ for $s=1, \dots, j-1$, but
$m^j> \eps$.\label{tcase}
\end{itemize}
In cases (a) and (b), there is a violated equation in \eqref{LPi+1}, 
while in~(c) there are none (but there is a violated
inequality). Thus (a) and (b) can be handled using the ideas
discussed in the beginning of Section~\ref{counts}. Indeed, as argued
there, we can efficiently compute the basis of the optimal face of
the feasible set of $\LP^{j-1}$ using the ellipsoid method.
One can then easily check if a candidate solution violates
one of the equations in the basis (recall that there are at most $n$ of them), 
and, if this is the case, 
report one that is violated. 
Hence, we only need to show how to identify a 
violated constraint in case~(c). However, for completeness, we present here a 
purely counting-based algorithm for each of the cases.

In case (a),  
let $(\hat{\p}, \eps^{\ell})$ be an interior optimizer
for $\LP^{\ell}$. 
Under $\hat{\p}$, the deficit of any coalition in 
$2^I\setminus\cup_{s=1}^{\ell-1}\Sigma^s$ is at most $\eps^\ell$.
On the other hand, under $\p$, there are $n^\ell$ coalitions
in $2^I\setminus\cup_{s=1}^{\ell-1}\Sigma^s$ 
whose deficit is $m^\ell>\eps^\ell$. Each of these coalitions
corresponds to a violated constraint: indeed, if such a coalition
is in $\Sigma^s$, $s\ge\ell$, then $\LP^j$ requires that its deficit
is $\eps^s\le\eps^\ell<m^\ell$, and if it is in 
$2^I\setminus\cup_{s=1}^{j-1}\Sigma^s$, then 
$\LP^j$ requires that its deficit
is at most $\eps\le\eps^\ell<m^\ell$,  
Hence, it suffices to identify a coalition whose deficit
under $\p$ is $m^\ell$. To this end, 
we can modify the dynamic program for $\p$
as follows. Together with every variable $X^t_{k, w}$,
$t=1, \dots, j$, $k=1, \dots, n$, $w=1, \dots, nW$, we will use
an auxiliary variable $Z^t_{k, w}$ which stores a coalition
whose payoff under $\p$ is equal to $X^t_{k, w}$. The values
of $Z^t_{k, w}$ can be easily computed by induction:
if  $X^t_{k, w}=X^s_{k-1, w}$ for some $s=1, \dots, j$
then $Z^t_{k, w}=Z^s_{k-1, w}$, and
if  $X^t_{k, w}=p_k+X^s_{k-1, w-w_k}$ for some $s=1, \dots, j$
then $Z^t_{k, w}=Z^s_{k-1, w}\cup\{k\}$
(if $X^t_{k, w}=X^s_{k-1, w}=p_k+X^{s'}_{k-1, w-w_k}$, 
we can set $Z^t_{k, w}$ to either of these values).
Now, there exist some $t$ and $w$ such that
$w\ge q$ and $1-X^t_{n, w}=m^\ell$ or
$w< q$ and $-X^t_{n, w}=m^\ell$; such $t$ and $w$ can be found
by scanning all $X^t_{n, w}$. The corresponding set $Z^t_{n, w}$
has deficit $m^\ell$ under $\p$, 
and therefore corresponds to a violated constraint.

In case (b), as before, 
let $(\p, \eps)$ be an interior optimizer to $\LP^{j-1}$.
There exists a coalition whose 
deficit under $\p$ is $\eps^\ell$, but whose deficit
under $\hat{\p}$ is strictly less than $\eps^\ell$.
To find such a coalition, 
run $\P(\hat{\p}, \eps)$
in order to compute the corresponding values
$\hat{X}^t_{k, w}$, $\hat{Y}^t_{k, w}$
$t=1, \dots, j-1$, $k=1, \dots, n$, $w=1, \dots, nW$.
Define $Z_w$ as follows:
if there exists some $t\in\{1, \dots, j-1\}$ such that
$X^t_{n, w}=I_w-\eps^\ell$, set $Z_w=Y^t_{n, w}$;
otherwise, set $Z_w=0$.
$\hat{Z}_w$ is defined similarly:
if there exists an $s\in\{1, \dots, j-1\}$ such that
$\hat{X}^s_{n, w}=I_w-\eps^\ell$, set $\hat{Z}_{w}=\hat{Y}^s_{n, w}$;
otherwise, set $\hat{Z}_{w}=0$.
The variables $Z_{w}$ and $\hat{Z}_{w}$ count the number of coalitions
that have total weight $w$ and have deficit $\eps^\ell$
under $\p$ and $\hat{\p}$, respectively.
We have 
$n^\ell =\sum_{w=1, \dots, nW}Z_{w}, 
 s^\ell =\sum_{w=1, \dots, nW}\hat{Z}_{w}$.
As $n^\ell>s^\ell$, there exists a weight $w$
such that $Z_w>\hat{Z}_w$. Set $q=I_w-\eps^\ell$, 
i.e., $q$ is the total payment received by the coalitions
counted by $Z_w$ and $\hat{Z}_w$.
Now, we have $Z_w=Z_w^n+Z_w^{-n}$, where
$Z_w^n$ is the number of coalitions of weight
$w$ that include $n$, have weight $w$ and receive total payment $q$, 
and 
$Z_w^{-n}$ is the number of coalitions of weight
$w$ that do not include $n$, have weight $w$ 
and receive total payment $\eps^\ell$;
$\hat{Z}^n_w$ and $\hat{Z}^{-n}_w$ can be defined similarly.
We can easily compute these quantities: for example, $Z^{n}_w$
is the number of subsets of $\{1, \dots, n-1\}$ that have weight
$w-w_n$ and receive total payment $q-p_n$, i.e.
$Z^{n}_w=Y^t_{n-1, w-w_n}$ if there exists a $t\in\{1, \dots, j-1\}$ such that
$X^t_{n-1, w-w_n}=q-p_n$, and $Z^{-n}_w=0$ otherwise.
It follows immediately that $Z^n_w>\hat{Z}^n_w$
or $Z^{-n}_w>\hat{Z}^{-n}_w$ (or both), and we can easily verify which of
these cases holds. In the former case, we can conclude that the number of coalitions
in $\{1, \dots, n-1\}$ 
that have weight $w-w_n$ and are paid $q-p_n$ under $\p$
exceeds the number of coalitions in $\{1, \dots, n-1\}$
that have weight $w-w_n$ and are paid $q-\hat{p}_n$ under $\hat{\p}$.
In the latter case, we can conclude that the number of coalitions
in $\{1, \dots, n-1\}$
that have weight $w$ and are paid $q$ under $\p$
exceeds the number of coalitions in $\{1, \dots, n-1\}$
that have weight $w$ and are paid $q$ under $\hat{\p}$.
Continuing in the same manner for $n-1, \dots, 1$, we can identify
a coalition that is paid $q$ under $\p$, but not under $\hat{p}$.

Case (c), i.e. $m^j>\eps$, is similar to (a) and can be handled in
the same manner.
 
\section{Conclusions and future work}\label{sec:conclusion} 
In this paper, we
proposed a new technique for computing the nucleolus of coalitional games.
Namely, we have shown that, when constructing the separation oracle for the
$j$th linear program $\LP^j$, instead of storing the sets of tight constraints
for the linear programs $\LP^t$, $t=1, \dots, j-1$, it suffices to store the
sizes of these sets as well as the top $j-1$ deficits of an interior optimizer
$(\p^{j-1}, \eps)$ to $\LP^{j-1}$.  A feasibility of a candidate solution $(\p,
\eps)$ to $\LP^j$ can then be verified, roughly, by computing the top $j$
deficits for $\p$ as well as the number of coalitions that have these deficits,
and comparing these values to their pre-computed counterparts for $(\p^{j-1},
\eps)$.

We then demonstrated the usefulness of this technique by showing that for
weighted voting games with polynomially-bounded weights both the top $j$
deficits and the number of coalitions that have these deficits can be
efficiently computed using dynamic programming. This allows us to implement the
separation oracles for our linear programs in pseudopolynomial time. Combining
this with the ellipsoid algorithm results in a pseudopolynomial time algorithm
for the nucleolus of weighted voting games, thus solving an open problem posed
by~\cite{\aaai}. Furthermore, the general technique 
put forward in this paper effectively reduces the computation of the nucleolus
to solving a natural combinatorial problem for the underlying game.
Namely, we can state the following meta-theorem:
\begin{theorem}
Given a coalitional game $G$, 
suppose that we can, for any payoff vector $\p$, identify the top $n$ distinct
deficits under $\p$ as well as the number of coalitions that have these deficits
in polynomial time.
Then we can compute the nucleolus of $G$ in polynomial time. 
\end{theorem}
We believe that this framework can be useful for computing the nucleolus
in other classes of games. Indeed, by stripping away most of the game-theoretic
terminology, we may be able to find the nucleolus by 
applying existing results in combinatorics and discrete mathematics 
in a black-box fashion.

In the context of weighted voting games, 
our assumption that the weights are polynomially bounded
(or, equivalently, given in unary) is essential, as~\cite{\aaai} shows
that the nucleolus is NP-hard to compute for WVGs with
weights given in binary. Moreover, in many practical scenarios the
agents' weights cannot be too large (e.g., polynomial functions of $n$), in
which case the running time of our algorithm is polynomial.

By a slight modification of our algorithm, we can obtain a pseudopolynomial
time algorithm for computing the nucleolus in {\em $k$-vector weighted voting
games} for constant $k$.  Informally speaking, these are games given by the
intersection of $k$ weighted voting games, i.e., a coalition is considered to
be winning if it wins in each of the underlying games. There are some
interesting games that can be represented as $k$-vector weighted voting games
for small values of $k$ (i.e., $k=2$ or $k=3$), but not as weighted voting
games, most notably, voting in the European Union~\cite{bilbao:2002a}. Hence,
this extension of our algorithm enables us to compute the nucleolus in some
real-life scenarios.  The overall structure of our algorithm remains the same.
However, the dynamic program has to be modified to keep track of several weight
systems simultaneously.

Another natural way to address the problem of computing the nucleolus is by
focusing on approximate solutions. Indeed, \cite{\aaai} proposes a fully
polynomial time approximation scheme (FPTAS) for several least-core related
problems. It would be natural to expect a similar result to hold for the
nucleolus.  Unfortunately, this approach is ruled out by~\cite{\aaai}, which
shows that it is NP-hard to decide whether the nucleolus payoff of any
particular player is 0, and therefore approximating the nucleolus payoffs up to
any constant factor is NP-hard. Nevertheless, one can attempt to find an {\em
additive} approximation to the nucleolus, i.e., for a given error bound
$\delta>0$, find a vector $\x$ such that $|\eta_i-x_i|\le\delta$ for $i=1,
\dots, n$.  This can be useful in situations when the agents' weights cannot be
assumed to be polynomially bounded with respect to $n$, e.g., in the multiagent
settings where the weights correspond to agents' resources.  We are currently
investigating several approaches to designing additive approximation algorithms
for the nucleolus.

\bibliographystyle{abbrv}
\bibliography{general,core-nucl}

\begin{contact}
Edith Elkind \\
Intelligence, Agents, Multimedia group \\
School of Electronics and Computer Science \\
University of Southampton\\
Southampton, SO17 1BJ, United Kingdom\\
\email{ee@ecs.soton.ac.uk}
\end{contact}

\begin{contact}
Dmitrii Pasechnik \\
Division of Mathematical Sciences\\
School of Physical and Mathematical Sciences\\
Nanyang Technological University\\
21 Nanyang Link\\
Singapore 637371\\
\email{dima@ntu.edu.sg}
\end{contact}

\end{document}